\title{Homophilic networks evolving by mimesis}
\author{
        Jos\'e Manuel Rodr\'iguez Caballero \\
                Institute of Computer Science \\ University of Tartu \\ Tartu, Estonia
}
\date{\today}

\documentclass[12pt]{article}

\usepackage[english]{babel}

\usepackage{blindtext}

\usepackage{amsmath}
\usepackage{hyperref}
\usepackage{lipsum}
\usepackage{amsfonts}
\usepackage{enumitem}
\usepackage{epigraph}
\usepackage{graphicx}
\usepackage{amsthm}

\newtheorem{thm}{Theorem}

\newtheorem{lem}{Lemma}

\theoremstyle{definition}
\newtheorem{defn}{Definition}

\begin{document}
\maketitle

\begin{abstract}
We provide a mathematical model for networks based on similarities (homophily) and evolving by mutual imitation (mimesis). We show that such social networks will converge to a state of segregation, where the in-group interactions will be maximal and there will be no out-group flow of information. We establish some connections between our model and the Wolfram model for fundamental physics.
\end{abstract}

\section{Introduction}

\epigraph{There is in [people's] natural
propensity, from childhood
onward, to engage in mimetic
activities.

And this distinguishes [people] from
other creatures, that [they are]
thoroughly mimetic and through
mimesis [take they] first steps in
understanding}{Aristotle (The
Poetics)}

The mimetic behavior in human activity is a well-established scientific fact \cite{garrels2011mimesis}. The origins of mimetic theory can be found in Plato's \emph{Republic} and Aristotle's \emph{Poetics} \cite{lawtoo2018violence}. R. Girald \cite{girard2014mensonge, girard1998mimesis} developed the theory of \emph{d\'esir mim\'etique}, in which \emph{mimesis} of the wishes is the main driving force of human social behavior. His main postulate is\footnote{A literary translation is: ``Man always desires according to the desire of the Other". Of course, ``Man" here stands for human, including any gender.}

\begin{quotation}
L'homme d\'esire toujours selon le d\'esir de l'Autre.
\end{quotation}

The theory of \emph{d\'esir mim\'etique}, which is mainly philosophical and literary \cite{deguy1982rene, lawtoo2018violence}, may be also related to well-studied neurological structures known as \emph{mirror neurons} \cite{ferrari2014mirror, ferrari2009monkey, heyes2010mirror, iacoboni2009imitation, williams2001imitation}. A mirror neuron in a human is a neuron that is activated by observing the behavior of other humans. Nevertheless, the connection between mirror neurons and imitation is still a subject of debate among scientists \cite{hickok2009eight}.

Mimetic human behavior is present in the theory of social laser \cite{khrennikov2016social}, where people willing to imitate each other are described in a similar way to the bosons from particle physics. In this approach, mimesis is modeled by a social version of Bose-Einstein statistics.

The tendency of humans to interact more frequently with other humans sharing similarities is known as \emph{homophily} \cite{mcpherson2001birds, bramoulle2012homophily, kim2017effect}. This tendency is also present in other animals and it is likely to play some evolutionary role \cite{fu2012evolution}. The homophilic network \cite[page 87]{meir2018strategic} is a representation of social connections among people are determined by proximity of interest, status, wealth, ethnicity, position in a hierarchy, etc. Several mathematical techniques, mostly from graph theory, have been developed in order to study this structure \cite[section 7.13]{newman2010networks}.

The aim of the present paper is to introduce a mathematical model for homophilic networks (Definition \ref{homophilicnetwork}) evolving by mimesis (equation \eqref{integralequation}). Our model is a generalization of the model for diffusion in graphs \cite[section 6.13.1]{newman2010networks}. Diffusion equations (equation \eqref{differentialequation}) have already been used in order to model social phenomena\cite{yurevich2018modeling}. Nevertheless, we have not found our equation \eqref{integralequation} in the existing literature, and the continuous update of the time-dependent network $t\mapsto G_{\varepsilon}\left( \boldsymbol{\psi}(t) \right)$ seems to be a new feature in the subject of diffusion models.

Making an analogy with fundamental physics via the Wolfram model \cite{wolfram2020class, gorard2020Relativistic, gorard2020Quantum}, we will show (Theorem \ref{limitbehavior}) that the evolution of a network, in our model, will converge to a state that is the social analogous of what cosmologist call the Black Hole Era, i.e., society will be fragmented into groups, with maximal in-group interaction and no out-group flow of information.

\section{Mathematical model}

Our standpoint is the following discrete structure, modeling homophily.

\begin{defn}\label{homophilicnetwork}
Fix a positive integer $n$ (size of the population) and an extended nonnegative real number\footnote{A number that is either a nonnegative real number or the positive infinite.} $\varepsilon \in [0,+\infty]$ (tolerance threshold). Consider a function that assigns an undirected simple graph $G_{\varepsilon}(\boldsymbol{\psi})$ to any $\boldsymbol{\psi} = \left(\psi_1, ..., \psi_n\right) \in \mathbb{R}^n$. The vertices of $G_{\varepsilon}(\boldsymbol{\psi})$ are the numbers $1, 2, ..., n$. In $G_{\varepsilon}(\boldsymbol{\psi})$, there is an edge between $i$ and $j$ if and only if $i \neq j$ and $|\psi_i - \psi_j| \leq \varepsilon$. We will call $G_{\varepsilon}(\boldsymbol{\psi})$ the \emph{homophilic network} of the population $\boldsymbol{\psi}$ at tolerance level $\varepsilon$.
\end{defn}

 In the next definition we provide a mathematical model for the evolution by mimesis of a homophilic network. We will use the notation $\boldsymbol{L}\left[ G\right]$ for the Laplacian matrix \cite[section 6.13.1]{newman2010networks} of the graph $G$.

\begin{defn}
Fix a positive integer $n$, a positive real number $C$ (diffusion constant), a real number $t_0$ (initial time) and a vector $\boldsymbol{\psi}_0\in\mathbb{R}^n$ (initial microstate). The \emph{homophilic-mimetic model} for the evolution of a network is given by the matrix integral equation\footnote{Notice that $\boldsymbol{\psi} \mapsto \boldsymbol{L}\left[ G_{\varepsilon}\left(\boldsymbol{\psi}\right) \right]$ is piecewise constant. So, the integral is well-defined.}
\begin{equation}\label{integralequation}
\boldsymbol{\psi}(t) + C \int_{t_0}^t \boldsymbol{L}\left[ G_{\varepsilon}\left(\boldsymbol{\psi}(\tau) \right) \right]\boldsymbol{\psi}(\tau) \, d\tau = \boldsymbol{\psi}_0,
\end{equation}
defined for $t \in (t_0, +\infty)$, where $\boldsymbol{\psi}: [t_0,+\infty] \longrightarrow \mathbb{R}^n$ is a continuous function\footnote{The existence of $\boldsymbol{\psi}(+\infty)$ is a consequence of Theorem \ref{limitbehavior} and the convergence of a diffusion process described by \eqref{differentialequation}.}.  We call $\boldsymbol{\psi}(t)$ the \emph{microstate of the network} at time $t$. The graph $G_{\varepsilon}\left(\boldsymbol{\psi}(t) \right)$ will be called the \emph{macrostate of the network} at time $t$. We will call \eqref{integralequation} the \emph{homophilic-mimetic equation}.
\end{defn}

The deduction of the equation \eqref{integralequation} as a description of a sociological phenomenon is given by the assumption that the mimetic human behavior in a group is a sort of diffusion process. We take the simplest model of diffusion, namely the differential equation \eqref{explicitsolution}, deduced from \emph{Newton's law of cooling}. Furthermore, we let the network of human interactions be updated every instant in order to guarantee the exchange of information among people sharing similarities and destroying connections among people who are different enough. In order to avoid discontinuities, we express the resulting set of differential equations (defined of different intervals where the network is constant), as an integral. We do not pretend to give an accurate description of social reality with this equation, but just a toy model of social evolution under the assumptions of homophily and mimesis. 

The macrostate of society $G_{\varepsilon}\left(\boldsymbol{\psi}(t) \right)$ can be easily measured using metadata from internet, e.g., online social networks, or even polls. On the other hand, the measurement of the microstate of society  $\boldsymbol{\psi}(t)$, and its unit of measurement, is an extremely complicated problem and it will not be discussed in the present paper. It is natural to guess that the microstate may be related to wealth, ethnicity, sexual attractiveness, social status, academic and military hierarchy, etc. Nevertheless, the choice of a random value producing an observed macrostate could to be the best strategy for simulations of real-life situations.

\section{Limit behavior}

Following S. Wolfram \cite{wolfram2020class} and J. Gorard \cite{gorard2020Relativistic, gorard2020Quantum}, we will interpret the evolution of the time-dependent network $t \mapsto G_{\varepsilon}\left(\boldsymbol{\psi}(t) \right)$ as the evolution of the spatial graph\footnote{The Wolfram model is really about \emph{hypergraphs}, but we will focus on the particular case of graphs.}, i.e., the evolution of space in the Wolfram model. In this framework it is natural to import concepts from fundamental physics, e.g., event horizon and black hole, to graph theory.

\begin{defn}
Consider a time-dependent network\footnote{In our model, $\mathcal{N}\left(t \right) = G_{\varepsilon}\left(\boldsymbol{\psi}(t) \right)$.} $t \mapsto \mathcal{N}\left(t \right)$, having the same set of vertices $V$ for all values of $t$. Let $W \subseteq V$ be a subset of vertices of the network . We say that $W$ will produce an \emph{event horizon} of $t \mapsto \mathcal{N}\left(t \right)$ if for all $t$ large enough, given two vertices, $w \in W$ and $z \in V \backslash W$ (complement of $W$), there is no path between $w$ and $z$ in $\mathcal{N}\left(t \right)$. Furthermore, if the induced subgraph\footnote{Let's recall that the \emph{induced subgraph} of a set of vertices $W$ is the graph obtained by restricting the graph to the vertices of $W$ and the edges to the edges having vertices of $W$ at both extremes.} of $W$ in $\mathcal{N}\left(t \right)$ is a complete graph\footnote{This implies that $W$ will evolve to an event horizon.} for all $t$ large enough, we say that $W$ will end as a \emph{black hole}. If there is a partition of $V$ such that each part will end as a black hole, then we say that the network $t \mapsto \mathcal{N}\left(t \right)$ will end as a \emph{Black Hole Era}\footnote{This name came from cosmology \cite{adams2016five}.}.
\end{defn}

The motivation behind this graph-theoretical definition of event horizon is that in physics it is a boundary in spacetime preventing any exchange of information between both sides of the boundary. In our model, the \emph{flow of information} is interpreted as the paths in the graph.

The analogy between black holes and complete graphs is given by the fact that complete graphs are the densest graphs\footnote{Let's recall that the \emph{density} of a graph is equal to the number of edges divided by the maximum possible number of edges.} whereas black holes are the densest forms of organized matter \cite{susskind2008black}.

The following theorem states that, in our model, networks will always end as a Black Hole Era.

\begin{thm}\label{limitbehavior}
Fix a positive integer $n$, an extended nonnegative real number $\varepsilon$, a positive real number $C$, a real number $t_0$ and a vector  $\boldsymbol{\psi}_0 \in \mathbb{R}^n$. The time-dependent network $t \mapsto G_{\varepsilon}\left(\boldsymbol{\psi}(t) \right)$ determined by \eqref{integralequation} will end as a Black Hole Era, i.e., there is a graph $\Omega$, which is a union of complete graphs and satisfies $G_{\varepsilon}\left( \boldsymbol{\psi}(t) \right) = \Omega$ for all $t$ large enough.
\end{thm}

As a preliminary, we need the following lemmas.

\begin{lem}\label{normnonincreasing}
Fix a positive integer $n$ and a graph $G$ having $n$ vertices. For any pair of real numbers $t_0 < t_1$, any positive real number $C$, any vector $\boldsymbol{\psi}_0 \in \mathbb{R}^n$ and any continuous function, differentiable on $(t_0, t_1)$, $\boldsymbol{\psi}: [t_0, t_1] \longrightarrow \mathbb{R}^n$ satisfying the matrix differential equation
\begin{equation}\label{differentialequation}
\frac{d \boldsymbol{\psi}}{dt} + C \boldsymbol{L}[G]\boldsymbol{\psi} = 0,
\end{equation}
on the interval $(t_0, t_1)$ and the initial condition $\boldsymbol{\psi}(t_0) = \boldsymbol{\psi}_0$, the inequality 
\begin{equation}\label{limitzero}
| \boldsymbol{\psi}(t) | \leq | \boldsymbol{\psi}(t_0) |
\end{equation}
holds for $t_0 \leq t \leq t_1$.
\end{lem}

\begin{proof}
Let $\boldsymbol{v}_1, \boldsymbol{v}_2,...,\boldsymbol{v}_n$ be an orthonormal system of eigenvectors of $\boldsymbol{L}[G]$, associated to the eigenvalues $0 = \lambda_1 \leq \lambda_2 \leq ... \leq \lambda_n$, respectively. The explicit solution of \eqref{differentialequation} satisfying the initial condition and continuous on $[t_0, t_1]$, is
\begin{equation}\label{explicitsolution}
\boldsymbol{\psi}(t) = \sum_{i=1}^n \langle \boldsymbol{\psi}_0 , \boldsymbol{v}_i\rangle e^{-C \lambda_i (t-t_0)} \boldsymbol{v}_i,
\end{equation}
for $t_0 \leq t \leq t_1$.

Applying the Pythagorean theorem,
\begin{equation}\label{explicitsolution}
|\boldsymbol{\psi}(t)|^2 = \sum_{i=1}^n |\langle\boldsymbol{\psi}_0 , \boldsymbol{v}_i\rangle|^2 e^{-2C \lambda_i (t-t_0)}.
\end{equation}

Also, using the fact that $t \mapsto |\langle \boldsymbol{\psi}_0 , \boldsymbol{v}_i\rangle|^2 e^{-2C \lambda_i (t-t_0)}$ is monotonically decreasing,
$$
\sum_{i=1}^n |\langle \boldsymbol{\psi}_0 , \boldsymbol{v}_i\rangle|^2 e^{-2C \lambda_i (t-t_0)} \leq \sum_{i=1}^n |\langle \boldsymbol{\psi}_0 , \boldsymbol{v}_i\rangle|^2.
$$

Notice that $\sum_{i=1}^n |\langle \boldsymbol{\psi}_0 , \boldsymbol{v}_i\rangle|^2 = |\boldsymbol{\psi}_0|^2$. Finally, by transitivity we get $|\boldsymbol{\psi}(t)|^2 \leq |\boldsymbol{\psi}_0|^2$. Therefore, $|\boldsymbol{\psi}(t)| \leq |\boldsymbol{\psi}_0|$.
\end{proof}

\begin{lem}\label{diameterbound}
Fix a positive integer $n$ and a graph $G$ having $n$ vertices. For any pair of real numbers $t_0 < t_1$, any positive real number $C$, any vector $\boldsymbol{\psi}_0 \in \mathbb{R}^n$, any connected component $H$ of $G$ and any continuous function, differentiable on $(t_0, t_1)$, $\boldsymbol{\psi}: [t_0, t_1] \longrightarrow \mathbb{R}^n$ satisfying the matrix differential equation \eqref{differentialequation} on the interval $(t_0, t_1)$ and the initial condition $\boldsymbol{\psi}(t_0) = \boldsymbol{\psi}_0$, the inequality 
\begin{equation}\label{limitzero}
\max_{i\in H}\psi_i(t) - \min_{j\in H}\psi_j(t) \leq 2(n - k_G) | \boldsymbol{\psi}_0 | e^{-C \lambda_G (t-t_0)},
\end{equation}
holds for all $t_0 \leq t \leq t_1$, where $\boldsymbol{\psi} = \left(\psi_1,...,\psi_n \right)$, $k_G$ is the number of connected components of $G$ and $\lambda_G$ is the minimum nonzero eigenvalue of $\boldsymbol{L}[G]$. We used the notation $i\in H$ to express that $i$ is a vertex in $H$.
\end{lem}

\begin{proof}
Let $\boldsymbol{v}_1, \boldsymbol{v}_2,...,\boldsymbol{v}_n$ and $\lambda_1, \lambda_2, ..., \lambda_n$ be as in the proof of Lemma \ref{normnonincreasing}. Using the fact that $k_G$ is the dimension of the kernel of $\boldsymbol{L}[G]$, defining $\boldsymbol{\psi}_{\infty} = \sum_{i=1}^{k_G} \langle \boldsymbol{\psi}_0 , \boldsymbol{v}_i\rangle \boldsymbol{v}_i$, we have
\begin{equation}
\boldsymbol{\psi}(t) - \boldsymbol{\psi}_{\infty} = \sum_{i=k_G+1}^n \langle \boldsymbol{\psi}_0 , \boldsymbol{v}_i\rangle e^{-C \lambda_i (t-t_0)} \boldsymbol{v}_i,
\end{equation}
for $t_0 \leq t \leq t_1$.

Applying the triangle inequality, the Cauchy-Schwarz inequality and the fact that $\boldsymbol{v}_1, \boldsymbol{v}_2,...,\boldsymbol{v}_n$  are normalized, we obtain the inequality
\begin{equation}
| \boldsymbol{\psi}(t) - \boldsymbol{\psi}_{\infty} | \leq (n-k_G) | \boldsymbol{\psi}_0 | e^{-C \lambda_G (t-t_0)}
\end{equation}
for all $t_0 \leq t \leq t_1$. Let $H_1, H_2,...,H_{k_G}$ be the connected components of $G$. Up to rearrangement of the indices, for $1 \leq r \leq k$ we have\footnote{$\# H_r$ is the number of vertices in the connected component $H_r$.} $\boldsymbol{v}_r= \left(\# H_r\right)^{-1/2}\left( v_{r, 1}, v_{r, 2},...,v_{r, n} \right)$, with $v_{r, i} = 1$ if $i$ is a vertex of $H_r$ and $v_{r, i} = 0$ otherwise. Hence, for any connected component $H_r$, the inequality
\begin{equation}
| \boldsymbol{\psi}_i(t) - \langle \boldsymbol{\psi}_0 , \boldsymbol{v}_r\rangle | \leq (n-k_G) | \boldsymbol{\psi}_0 | e^{-C \lambda_G (t-t_0)}
\end{equation}
holds for all $i$ which are vertices of $H_r$. Finally, applying triangle inequality
\begin{eqnarray*}
& & \max_{i\in H_r}\psi_i(t) - \min_{j\in H_r}\psi_j(t) \\
&=& \left|\max_{i\in H_r}\psi_i(t) - \min_{j\in H_r}\psi_j(t) \right| \\
&\leq& \left|\max_{i\in H_r}\psi_i(t) - \langle \boldsymbol{\psi}_0 , \boldsymbol{v}_r\rangle\right|
 + \left| \langle \boldsymbol{\psi}_0 , \boldsymbol{v}_r\rangle - \min_{j\in H_r}\psi_j(t)  \right| \\
&\leq& 2(n-k_G) | \boldsymbol{\psi}_0 | e^{-C \lambda_G (t-t_0)}.
\end{eqnarray*}
\end{proof}

Now, we proceed to the proof of our main result.

\begin{proof}(of Theorem \ref{limitbehavior})
Let $t_0 < t_1 < t_2 <...$ be all the instants such that $G_{\varepsilon}\left(\boldsymbol{\psi}(t_{i+1}) \right)$ is not equal to $G_{\varepsilon}\left(\boldsymbol{\psi}(t_i) \right)$. In virtue of Lemma \ref{diameterbound}, there are positive real numbers $M_{G}$ and $\lambda_{G}$ such that,
for any connected component $H_k$ of $G_{\varepsilon}\left(\boldsymbol{\psi}(t_k) \right)$ the inequality\footnote{The notation $i\in H_k$ was explained in the statement of Lemma \ref{diameterbound}.}
\begin{equation}
\max_{i\in H_k}\psi_i(t) - \min_{j\in H_k}\psi_j(t)  \leq M_G | \boldsymbol{\psi}_k | e^{-C \lambda_G (t-t_0)}
\end{equation}
holds for $t_k \leq t < t_{k+1}$, where $G = G_{\varepsilon}\left(\boldsymbol{\psi}(t_k) \right)$. Also, the inequality
\begin{equation}
M_G | \boldsymbol{\psi}_k | e^{-C \lambda_G t} \leq M | \boldsymbol{\psi}_k | e^{-C \lambda (t-t_0)}
\end{equation}
holds, where $M$ is the maximum of $M_G$ for all $G$ having $n$ vertices, and $\lambda$ is the minimum of $\lambda_G$ for all $G$ having $n$ vertices. In virtue of Lemma \ref{normnonincreasing},
\begin{equation}
M | \boldsymbol{\psi}_k | e^{-C \lambda (t-t_0)} \leq M | \boldsymbol{\psi}_0 | e^{-C \lambda (t-t_0)}.
\end{equation}

Finally, we get
\begin{equation}
\max_{i\in H_k}\psi_i(t) - \min_{j\in H_k}\psi_j(t)  \leq M | \boldsymbol{\psi}_0 | e^{-C \lambda (t-t_0)},
\end{equation}
for all $t_k \leq t < t_{k+1}$.

Using the fact that the right hand side of the equation above does not depend on $k$, we conclude that, for all $t \geq t_0$, and any connected component $H_t$ of $G_{\varepsilon}\left(\boldsymbol{\psi}(t) \right)$, the inequality
\begin{equation}
\max_{i\in H_t}\psi_i(t) - \min_{j\in H_t}\psi_j(t)  \leq M | \boldsymbol{\psi}_0 | e^{-C \lambda (t-t_0)}.
\end{equation}
holds.  So, for all $t$ large enough and for any connected component $H_t$ of $G_{\varepsilon}\left(\boldsymbol{\psi}(t) \right)$ we have
\begin{equation}
\max_{i\in H_t}\psi_i(t) - \min_{j\in H_t}\psi_j(t)  \leq \varepsilon.
\end{equation}

According to Definition \ref{homophilicnetwork} and the inequality above, any connected component of $G_{\varepsilon}\left(\boldsymbol{\psi}(t) \right)$  will be a complete graph and this structure will not change in the future. Therefore, there is a graph $\Omega$, which is the union of complete graphs, such that, for all $t$ large enough, $G_{\varepsilon}\left(\boldsymbol{\psi}(t) \right) = \Omega$.
\end{proof}

\section{Final remarks}
Needless to say that our model is not meant to make numerical predictions about society. What we developed was an idealization. The applications of our model could be in order to explore some sociological theories using the intuition from our model, i.e., to compare how claims about the human behavior in society agree or disagree with the properties of our model. At the moment of writing this paper, it is not clear to what extent our model is related to actual statistical data of real-life social networks. Also, it may be interesting to look for applications of our model in physics, chemistry, biology and computer science.

\section*{Acknowledgments and disclaimers}
The author is an external affiliate of the Wolfram Physics Project but is not engaged by any formal agreement in any activity constituting a competition with his Employer (University of Tartu). The production of the present paper, done during the free time of the author, was only the result of intellectual curiosity and it was not funded by any company or institution.

The author would like to thank Matthew Szudzik for the suggestions in order to improve the mathematica notebook ``A Wolfram-like model of language secessionism".

\bibliographystyle{alpha}
\bibliography{mybibfile}

\newcommand{\etalchar}[1]{$^{#1}$}
\begin{thebibliography}{WWSP01}

\bibitem[AL16]{adams2016five}
Fred~C Adams and Greg Laughlin.
\newblock {\em The five ages of the universe: inside the physics of eternity}.
\newblock Simon and Schuster, 2016.

\bibitem[BCJ{\etalchar{+}}12]{bramoulle2012homophily}
Yann Bramoull{\'e}, Sergio Currarini, Matthew~O Jackson, Paolo Pin, and Brian~W
  Rogers.
\newblock Homophily and long-run integration in social networks.
\newblock {\em Journal of Economic Theory}, 147(5):1754--1786, 2012.

\bibitem[DD82]{deguy1982rene}
Michel Deguy and Jean-Pierre Dupuy.
\newblock {\em Ren{\'e} {G}irard et le probl{\`e}me du mal}.
\newblock Grasset, 1982.

\bibitem[FBF09]{ferrari2009monkey}
PF~Ferrari, L~Bonini, and L~Fogassi.
\newblock From monkey mirror neurons to primate behaviours: possible
  ‘direct’and ‘indirect’ pathways.
\newblock {\em Philosophical Transactions of the Royal Society B: Biological
  Sciences}, 364(1528):2311--2323, 2009.

\bibitem[FNCF12]{fu2012evolution}
Feng Fu, Martin~A Nowak, Nicholas~A Christakis, and James~H Fowler.
\newblock The evolution of homophily.
\newblock {\em Scientific reports}, 2:845, 2012.

\bibitem[FR14]{ferrari2014mirror}
Pier~Francesco Ferrari and Giacomo Rizzolatti.
\newblock Mirror neuron research: the past and the future, 2014.

\bibitem[Gar11]{garrels2011mimesis}
Scott~R Garrels.
\newblock {\em Mimesis and {S}cience: {E}mpirical {R}esearch on {I}mitation and
  the {M}imetic {T}heory of {C}ulture and {R}eligion}.
\newblock MSU Press, 2011.

\bibitem[Gir98]{girard1998mimesis}
Ren{\'e} Girard.
\newblock Mimesis and violence.
\newblock {\em Herm{\`e}s, La Revue}, 1(22):47--52, 1998.

\bibitem[Gir14]{girard2014mensonge}
Ren{\'e} Girard.
\newblock {\em Mensonge romantique et v{\'e}rit{\'e} romanesque}.
\newblock Grasset, 2014.

\bibitem[Gor20a]{gorard2020Quantum}
Jonathan Gorard.
\newblock {S}ome {Q}uantum {M}echanical {P}roperties of the {W}olfram {M}odel,
  2020.

\bibitem[Gor20b]{gorard2020Relativistic}
Jonathan Gorard.
\newblock {S}ome {R}elativistic and {G}ravitational {P}roperties of the
  {W}olfram {M}odel.
\newblock {\em arXiv preprint arXiv:2004.14810}, 2020.

\bibitem[Hey10]{heyes2010mirror}
Cecilia Heyes.
\newblock Where do mirror neurons come from?
\newblock {\em Neuroscience \& Biobehavioral Reviews}, 34(4):575--583, 2010.

\bibitem[Hic09]{hickok2009eight}
Gregory Hickok.
\newblock Eight problems for the mirror neuron theory of action understanding
  in monkeys and humans.
\newblock {\em Journal of cognitive neuroscience}, 21(7):1229--1243, 2009.

\bibitem[Iac09]{iacoboni2009imitation}
Marco Iacoboni.
\newblock Imitation, empathy, and mirror neurons.
\newblock {\em Annual review of psychology}, 60:653--670, 2009.

\bibitem[KA17]{kim2017effect}
Kibae Kim and J{\"o}rn Altmann.
\newblock Effect of homophily on network formation.
\newblock {\em Communications in Nonlinear Science and Numerical Simulation},
  44:482--494, 2017.

\bibitem[Khr16]{khrennikov2016social}
Andrei Khrennikov.
\newblock ‘{S}ocial laser’: action amplification by stimulated emission of
  social energy.
\newblock {\em Philosophical Transactions of the Royal Society A: Mathematical,
  Physical and Engineering Sciences}, 374(2058):20150094, 2016.

\bibitem[Law18]{lawtoo2018violence}
Nidesh Lawtoo.
\newblock Violence and the mimetic unconscious (part one): The cathartic
  hypothesis: Aristotle, freud, girard.
\newblock {\em Contagion}, 25:159--192, 2018.

\bibitem[Mei18]{meir2018strategic}
Reshef Meir.
\newblock Strategic voting.
\newblock {\em Synthesis Lectures on Artificial Intelligence and Machine
  Learning}, 13(1):1--167, 2018.

\bibitem[MSLC01]{mcpherson2001birds}
Miller McPherson, Lynn Smith-Lovin, and James~M Cook.
\newblock Birds of a feather: Homophily in social networks.
\newblock {\em Annual review of sociology}, 27(1):415--444, 2001.

\bibitem[New10]{newman2010networks}
Mark E.~J. Newman.
\newblock {\em Networks: {A}n {I}ntroduction}.
\newblock Oxford Univertity Press, 2010.

\bibitem[Sus08]{susskind2008black}
Leonard Susskind.
\newblock {\em The black hole war: My battle with {S}tephen {H}awking to make
  the world safe for quantum mechanics}.
\newblock Hachette UK, 2008.

\bibitem[Wol20]{wolfram2020class}
Stephen Wolfram.
\newblock A {C}lass of {M}odels with the {P}otential to {R}epresent
  {F}undamental {P}hysics.
\newblock {\em arXiv}, pages arXiv--2004, 2020.

\bibitem[WWSP01]{williams2001imitation}
Justin~HG Williams, Andrew Whiten, Thomas Suddendorf, and David~I Perrett.
\newblock Imitation, mirror neurons and autism.
\newblock {\em Neuroscience \& Biobehavioral Reviews}, 25(4):287--295, 2001.

\bibitem[YOMV18]{yurevich2018modeling}
Petukhov~Alexander Yurevich, Malkhanov~Alexey Olegovich, Sandalov~Vladimir
  Mikhailovich, and Petukhov~Yuri Vasilievich.
\newblock Modeling conflict in a social system using diffusion equations.
\newblock {\em Simulation}, 94(12):1053--1061, 2018.

\end{thebibliography}

\end{document}